\documentclass{article}

\usepackage{xcolor}
\definecolor{red}{RGB}{255,0,0}
\definecolor{blue}{rgb}{0.0, 0.4, 0.65}
\definecolor{orange}{RGB}{253,188,64}
\definecolor{green}{RGB}{154,205,50}

\usepackage{fullpage}
\usepackage[
colorlinks,
linkcolor=blue,
citecolor=blue,
urlcolor=blue,
]{hyperref}
\usepackage{authblk}
\usepackage[shortlabels]{enumitem}
\usepackage{csquotes} 
\providecommand{\keywords}[1] 
{
	\noindent \small	
	\textbf{\textit{Keywords---}} #1
}

\usepackage{amsmath,amssymb}
\allowdisplaybreaks 

\newcommand\N{\mathbb{N}}

\newcommand\I{{\mathcal{I}}}

\usepackage{amsthm}

\newtheorem{theorem}{Theorem}
\newtheorem{lemma}[theorem]{Lemma}

\newtheorem{proposition}[theorem]{Proposition}

\theoremstyle{definition}

\theoremstyle{remark}
\newtheorem{remark}{Remark}

\usepackage{graphics,graphicx,caption,subfig}
\graphicspath{{figures/}}

\usepackage{tikz}
\usetikzlibrary{calc,shapes}
\usetikzlibrary{shapes.multipart} 

\tikzset{
	every node/.style = {align=center},
	drawing/.style = {inner sep=0, outer sep=0},
	class/.style = {draw, minimum size=.7cm},
	server/.style = {draw, circle, minimum size=.85cm},
	token/.style = {draw, minimum size=.4cm},
	fcfs/.style={
		draw,
		rectangle split,
		rectangle split parts=#1,
		rectangle split horizontal,
		rectangle split empty part width=-.2cm,
		rectangle split empty part height=.2cm,
		inner ysep=.27cm,
	},
	ps/.style={
		rotate=90,
		draw,
		rectangle split,
		rectangle split parts=#1,
		rectangle split horizontal,
		rectangle split empty part width=-.17cm,
		rectangle split empty part height=.9cm,
	},
	serverps/.style = {draw, circle, minimum size=1cm},
}

\usepackage{pgfplots,pgfplotstable}
\pgfplotsset{compat=1.8}
\usepgfplotslibrary{fillbetween}

\pgfplotsset{
	table/col sep = {comma},
	defaultplotstyle/.style={
		every axis plot/.append style={
			thick,
		},
		xlabel near ticks, ylabel near ticks,
		ylabel style={align=center},
		xtick style={draw=none},
		ytick style={draw=none},
		grid=major,
		legend style={
			cells={anchor=west, align=left},
			at={(0.5, .98)},
			anchor=north,
			/tikz/every even column/.append style={column sep=0.2cm},
			font=\footnotesize,
		},
		legend cell align={left},
		width=.7\linewidth, height=.27\linewidth,
	},
	legendplotstyle/.style={defaultplotstyle,
		xmin=0, xmax=1,
		ymin=0, ymax=1,
		hide axis,
		legend columns=4,
		every axis plot/.append style={
			very thick,
		},
	},
}

\usepackage{filecontents}

\begin{document}

\title{Load Balancing in Heterogeneous Server Clusters: Insights~From~a~Product-Form~Queueing~Model\footnote{Author version of the paper available at \url{https://doi.org/10.1109/IWQOS52092.2021.9521355}.}}

\author{Mark van der Boor}
\author{C\'eline Comte}
\affil{Eindhoven University of Technology}

\maketitle

\begin{abstract}
	Efficiently exploiting servers in data centers requires performance analysis methods that account not only for the stochastic nature of demand but also for server heterogeneity. Although several recent works proved optimality results for heterogeneity-aware variants of classical load-balancing algorithms in the many-server regime, we still lack a fundamental understanding of the impact of heterogeneity on performance in finite-size systems. In this paper, we consider a load-balancing algorithm that leads to a product-form queueing model and can therefore be analyzed exactly even when the number of servers is finite. We develop new analytical methods that exploit its product-form stationary distribution to understand the joint impact of the speeds and buffer lengths of servers on performance. These analytical results are supported and complemented by numerical evaluations that cover a large variety of scenarios.
\end{abstract}

\keywords{Load balancing, performance analysis, product-form queueing model, Jackson network, insensitivity}

\section{Introduction} \label{sec:introduction}

Distributing a stochastic demand
across a set of heterogeneous servers
is not only a fundamental problem in queueing theory
but also an essential building block
of applications like
parallel computing
and production systems.
Besides static approaches
that do not require communication
between the dispatcher and the servers,
classical solutions include
join-the-shortest-queue,
power-of-$d$-choices~\cite{M01},
and join-idle-queue~\cite{LXKGLG11}.
These solutions were originally
designed for service systems
that are \textit{homogeneous}
in the sense that all servers
have the same speed~\cite{BBLM18}.
Although these solutions successfully cope
with the stochastic nature of demand,
they are not always suitable
when servers have unequal speeds~\cite{GJWD20}.

Several heterogeneity-aware approaches
were introduced to improve
the performance of these classical solutions,
for instance by using information on
the job sizes and server speeds
or by delaying the assignment decision.
For instance, join-the-shortest-workload
yields optimal performance
if the service times of all jobs
over all servers are known,
an assumption that is rarely satisfied in practice.
Redundancy scheduling achieves
the same performance gain~\cite{ABV18},
this time by delaying
the assignment decision,
which may again be practically infeasible if
the communication time
between the dispatcher and servers
is not negligible.
To achieve good performance
without these strong assumptions,
more recent works introduced speed-aware variants
of the above-mentioned well-known algorithms.
More specifically,
\cite{WZS20} introduced variants of
join-the-shortest-queue and join-idle-queue
where the server speeds are used
as a tie-breaking rule,
and proved that these variants
minimize the mean response time
in the many-server regime;
\cite{GJWD20} proposed variants of
power-of-$d$-choices and join-idle-queue
for service systems with two server types
(fast and slow)
by adapting the degree of diversity
and assignment probabilities
to the server speeds,
and proved stability,
again in the many-server regime.
Despite these advances,
we still lack a fundamental understanding of
the impact of heterogeneity on performance
in service systems with
a finite number of servers.

In this paper, we make one step
further into this direction
by considering a load-balancing
algorithm \cite{BJP04,C19-1,C19-2,JP18}
that leads to a product-form
queueing model
(assuming that jobs arrive
according to a Poisson process),
and can therefore be analyzed exactly
even when the number of servers is finite.
Assuming that each server
has a finite-length buffer,
this algorithm assigns each incoming job
to a server chosen at random,
with a probability proportional
to the number of available slots
in the server's buffer;
an incoming job that finds
the buffers of all servers full
is rejected and considered permanently lost.
Although this algorithm
does not account for the server speeds
in the online assignment decision, these speeds
can be used
offline to adjust the buffer lengths.
Besides its analytical tractability,
this algorithm has the advantage of
making performance insensitive to
the job size distribution beyond its mean,
provided that servers
apply the processor-sharing policy.
This insensitivity property,
which contributed to
the success of the Erlang-B formula
for dimensioning circuit-switched networks,
guarantees that
the long-run performance metrics
are not impacted by fine-grained traffic characteristics.
In this paper, we use the product-form queueing model
to better understand the impact
of parameters on performance.

\paragraph*{Related work on insensitive load balancing}

The works \cite{AW12,BJP04,JM10}
analyze the performance of
variants and generalizations
of this load-balancing algorithm
in heterogeneous service systems
where jobs have constraints
that restrict their assignment to resources.
However, the objective of these works
is to develop methods
to \emph{calculate} performance metrics,
which is a different goal
than \emph{understanding}
the impact of parameters on performance.
Some of the formulas derived in these works
are used in the numerical-evaluation section
to assess our analytical results.
The objectives of
the related work~\cite{JP18}
are closer to our work.
This work analyzes the performance
of the same algorithm,
but it focuses on the many-server and heavy-traffic regimes and assumes, for the most part, that servers are homogeneous.
The product-form stationary distribution of load-balancing models has also been studied in~\cite{BBL17,C19-1,C19-2,GR20} for systems with multiple dispatcher or arbitrary server-job compatibilities (also see references therein). 

\paragraph*{Additional related work on load balancing}

If all servers have unit-length buffers,
our algorithm can be seen
as a loss variant
of join-idle-queue~\cite{LXKGLG11},
whereby a job is rejected
if all servers are busy upon its arrival.
With arbitrary buffer lengths,
our algorithm is related
to idle-one-queue~\cite{GW19} and join-below-threshold~\cite{GBLMW20,HSH19,ZTS18-1,ZTS18-2,ZWTSS17},
two generalizations of
join-idle-queue introduced
to improve performance
in the heavy-traffic regime
or when servers have unequal speeds.
The idea is that
servers notify the dispatcher
when the number of jobs in their buffer
falls below a threshold,
so that the dispatcher
assigns incoming jobs
to lightly-loaded servers
if possible.
The threshold of a server,
equal to one under join-idle-queue,
offers a trade-off between
performance improvement and communication overhead;
in case of unequal server speeds,
it can also be used
to favor faster servers.
In our algorithm,
these thresholds correspond to
the buffer lengths,
and the trade-off between
performance and communication overhead
is materialized by
the overall buffer length
(that is, the sum of the lengths
of the buffers at all servers).
Assuming that
this overall buffer length is fixed,
we would like to understand
how to optimally choose
the buffer length of each server
depending on the job arrival rate
and the server speeds.

\paragraph*{Contributions}

Our contributions can be summarized as follows.
We first show the following results
for a cluster of two servers
in which the overall buffer size
across the two servers is fixed.
When the arrival rate is low,
the optimal buffer lengths
(to minimize the loss probability)
are proportional to the server speeds,
meaning that the
buffer of a server is longer
if this server is faster.
On the contrary, when the arrival rate is large,
the optimal buffer lengths are uniform,
and the server speeds
only intervene to break ties
if the overall buffer length is odd.
We also show that,
between these two limiting regimes,
the optimal buffer
lengths evolve monotonically with
the arrival rate.
Besides the practical implications
of these results in terms of system design,
the analytical methods that we develop,
based on an analogy with
weighted paths in the two-dimensional lattice,
are of independent interest.
Afterwards,
we explain how these results
extend to clusters of more than two servers.
We finally turn to numerical evaluations
to assess the validity of these results
in practice
and understand
the impact of parameters
on the mean response time.

\paragraph*{Organization of the paper}

The remainder of the paper
is organized as follows.
Section~\ref{sec:model}
introduces the cluster model
and the equivalent closed queueing model.
In Section~\ref{sec:performance},
we use this model
to derive closed-form expressions
for several performance metrics,
such as the loss probability.
Sections~\ref{sec:underloaded}
to~\ref{sec:monotonicity}
contain our main contributions
for clusters of two servers.
In Section~\ref{sec:underloaded},
we prove that the optimal buffer lengths
in terms of the loss probability
are proportional to the server speeds
when the arrival rate tends to zero,
while, in Section~\ref{sec:overloaded},
we prove that
the optimal buffer lengths are uniform
when the arrival rate tends to infinity.
Section~\ref{sec:monotonicity}
fills the gap between
these two limiting regimes
by showing a monotonicity result.
These results are generalized
to clusters of more than two servers
in Section~\ref{sec:generalization}.
Section~\ref{sec:num}
gives numerical evaluations
and Section~\ref{sec:ccl} concludes the paper.

\section{Heterogeneous server cluster} \label{sec:model}

We consider a cluster that consists of
a single dispatcher and two servers.
Incoming jobs arrive at the dispatcher
according to a Poisson process
with rate $\lambda$.
Each server has a finite buffer,
of length $\ell_1$ for server~1
and $\ell_2$ for server~2,
that contains all jobs assigned to this server
(either waiting or in service),
with the total buffer size
being denoted by $L = \ell_1 + \ell_2$.
An incoming job is permanently lost
if the buffers of both servers
are full upon its arrival,
otherwise the dispatcher immediately
assigns the job to one of the servers,
as will be elaborated further.
Two equivalent state descriptors are
the vector $n = (n_1, n_2)$ that counts
the jobs in the buffer of each server,
and the vector $x = (x_1, x_2)$ that counts
the available slots in the buffers of each server.
The generalization of this model
to clusters of more than two servers
will be considered in
Section~\ref{sec:generalization}.

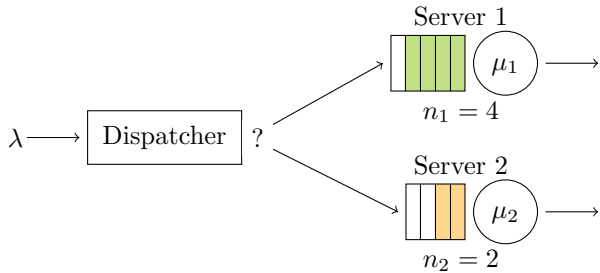
\begin{figure}[ht]
	\centering
	\begin{tikzpicture}
		\node[fcfs=5,
		rectangle split part fill
		={white, green!60}
		] (queue1) {};
		\node[server, anchor=west] (server1)
		at ($(queue1.east)+(.1cm,0)$) {};
		\node at (server1) {\strut$\mu_1$};
		\node[anchor=south]
		at ($(server1.north east)-(.9cm,0)+(0,.08cm)$)
		{Server~1};
		\node[anchor=north]
		at ($(server1.south east)-(.9cm,0)-(0,.08cm)$)
		{$n_1 = 4$};
		
		\node[fcfs=4,
		rectangle split part fill
		={white, white, orange!60},
		anchor=east
		] (queue2) at ($(queue1.south east)-(0,1.6cm)$) {};
		\node[server, anchor=west] (server2)
		at ($(queue2.east)+(.1cm,0)$) {};
		\node at (server2) {\strut$\mu_2$};
		\node[anchor=south]
		at ($(server2.north east)-(.9cm,0)+(0,.08cm)$)
		{Server~2};
		\node[anchor=north]
		at ($(server2.south east)-(.9cm,0)-(0,.08cm)$)
		{$n_2 = 2$};
		
		\node[draw, inner sep=.2cm] (dispatcher)
		at ($(queue1.east)!.5!(queue2.east)-(4cm,0)$)
		{Dispatcher};
		
		\draw[->]
		($(dispatcher.west)-(.8cm,0)$)
		-- node[near start, anchor=east,
		xshift=-.1cm] {$\lambda$}
		($(dispatcher.west)-(.1cm,0)$);
		\draw[->]
		($(dispatcher.east)+(.1cm,0)$)
		-- ($(queue1.west)-(.1cm,0)$);
		\draw[->]
		($(dispatcher.east)+(.1cm,0)$)
		-- ($(queue2.west)-(.1cm,0)$);
		\draw[->]
		($(server1.east)+(.1cm,0)$)
		-- ($(server1.east)+(.8cm,0)$);
		\draw[->]
		($(server2.east)+(.1cm,0)$)
		-- ($(server2.east)+(.8cm,0)$);
		
		\node[anchor=west, fill=white] (question)
		at ($(dispatcher.east)$) {?};
	\end{tikzpicture}
	\caption{A heterogeneous cluster
		with two servers.}
	\label{fig:cluster}
\end{figure}

\subsection{Load balancing and scheduling} \label{subsec:cluster-load-balancing}

The dispatcher applies
the following randomized
load-balancing algorithm,
considered in \cite{BJP04,C19-1,C19-2,JP18}.
When a new job arrives,
the dispatcher chooses
a server at random,
with a probability
proportional to the number
of available slots in the buffer of the server,
and assigns the job to this server.
In \figurename~\ref{fig:cluster}
for instance,
there are one and two
available slots in the buffers
of servers~1 and~2, respectively,
so that an incoming job would be assigned
to server~1 with probability~$\frac13$
and to server~2 with probability~$\frac23$.
In general, if there are
$x_1$ and $x_2$ available slots
in the buffers of servers~1 and~2,
respectively,
then server~1 is chosen
with probability
$\frac{x_1}{x_1 + x_2}$
and server~2 with probability
$\frac{x_2}{x_1 + x_2}$
if $x_1 + x_2 \ge 1$,
otherwise the job is lost.
This algorithm assumes that the dispatcher
always knows the number
of available slots
in the buffer of each server;
this happens, for instance, if
jobs go through the dispatcher
again when they leave the system,
or if servers notify the dispatcher
upon service completions.
This load-balancing algorithm
was considered
in \cite{BJP04,C19-1,C19-2,JP18}
for clusters with an arbitrary number of servers.

We assume that each server
processes the jobs in its buffer
according to a non-anticipating
work-conserving scheduling algorithm,
such as processor-sharing
or first-come-first-served.
The service rates of the servers,
assumed to be constant for simplicity,
are denoted by $\mu_1$ and $\mu_2$.
We assume without loss of generality
that $1 > \mu_1 > \mu_2 > 0$
and $\mu_1 + \mu_2 = 1$,
so that server~1 is the fastest.
The job service requirements are independent
and exponentially distributed with unit mean,
so that the remaining service time
of a job is exponentially distributed
with mean $\frac1\mu$
if this job is currently
served at rate~$\mu$.
This memoryless assumption
is actually not required
to perform the subsequent analysis
if each server applies processor-sharing
or preemptive-resume last-come-first-served.
Each job leaves the system immediately
upon service completion.

Even if the dispatcher does not take
the service rates
$\mu_1$ and $\mu_2$
into account when making
the assignment decision,
we will see later that
varying the buffer lengths
$\ell_1$ and $\ell_2$
allows us to optimize the load distribution
with respect to the service rates.

\subsection{Queueing model}
\label{subsec:cluster-queueing}

The dynamics can be described by
a closed Jackson network \cite{S99}
of three stations
with two customer (or token) classes
\cite{BBL17,C19-1}.
In Section~\ref{sec:performance},
we will use this observation
to derive closed-form expressions
for several performance metrics.

\begin{figure}[ht]
	\centering
	\begin{tikzpicture}
		\node[fcfs=5,
		rectangle split part fill
		={white, green!60}
		] (queue1) {};
		\node[server, anchor=west] (server1)
		at ($(queue1.east)+(.1cm,0)$) {};
		\node at (server1) {\strut$\mu_1$};
		\node[anchor=south]
		at ($(server1.north east)-(.9cm,0)+(0,.08cm)$)
		{Server-1 station};
		\node[anchor=north]
		at ($(server1.south east)-(.9cm,0)-(0,.08cm)$)
		{$n_1 = 4$};
		
		\node[fcfs=4,
		rectangle split part fill
		={white, white, orange!60},
		anchor=east
		] (queue2)
		at ($(queue1.south east)-(0,1.5cm)$) {};
		\node[server, anchor=west] (server2)
		at ($(queue2.east)+(.1cm,0)$) {};
		\node at (server2) {\strut$\mu_2$};
		\node[anchor=south]
		at ($(server2.north east)-(.9cm,0)+(0,.08cm)$)
		{Server-2 station};
		\node[anchor=north]
		at ($(server2.south east)-(.9cm,0)-(0,.08cm)$)
		{$n_2 = 2$};
		
		\node[fcfs=9,
		rectangle split part fill
		={orange!60, green!60, orange!60, white},
		anchor=east
		] (dispatcher)
		at ($(queue1.south east)!.5!(queue2.north east)
		-(4cm,0)$) {};
		\node[server, anchor=east] (lambda)
		at ($(dispatcher.west)-(.1cm,0)$) {};
		\node at (lambda) {\strut$\lambda$};
		\node[anchor=south]
		at (dispatcher.north)
		{Dispatcher station};
		\node[anchor=north]
		at (dispatcher.south)
		{$(x_1, x_2) = (1,2)$};
		
		\draw[<->]
		($(dispatcher.east)+(.1cm,0)+(0,.1cm)$)
		-- node[midway, fill=white, inner sep=.08cm]
		{Class~1} ($(queue1.west)-(.1cm,0)$);
		\draw[<->]
		($(dispatcher.east)+(.1cm,0)-(0,.1cm)$)
		-- node[pos=.46, fill=white, inner sep=.08cm]
		{Class~2} ($(queue2.west)-(.1cm,0)$);
	\end{tikzpicture}
	\caption{Jackson network
		associated with the cluster
		of \figurename~\ref{fig:cluster}.}
	\label{fig:jackson}
\end{figure}
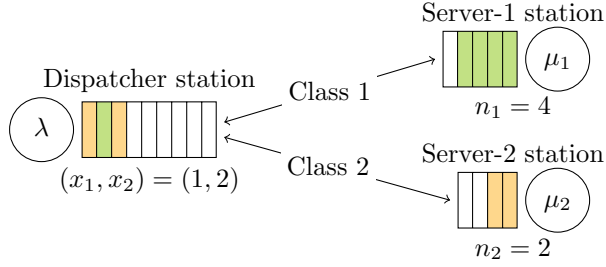

Instead of counting how many jobs are present in the cluster, we keep track to how tokens evolve in a network, where every token corresponds to a specific slot in a server's buffer. A token corresponding to a slot in the buffer of server~$i$ is said to be of class~$i$, for $i \in \{1,2\}$. When a slot in the buffer of server~$i$ is not occupied, the corresponding class-$i$ token is located at the \emph{dispatcher station}; when this slot becomes occupied by a job, the class-$i$ token moves to the \emph{server-$i$ station}.
The routing mechanism of tokens
is deterministic, and each class-$i$ token moves only between the dispatcher station and the server-$i$ station.
More specifically, whenever a job completes service at server~$i$ (this happens with rate $\mu_i$), the corresponding slot becomes available and the corresponding token moves to the dispatcher station. Whenever a job is assigned to server~$i$, a slot in the buffer of this server becomes occupied by this job, so that the corresponding class-$i$ token moves to server-$i$ station.

The corresponding closed Jackson network
consists of the dispatcher station
and the server-1 and 2 stations.
There are $\ell_i$ tokens of class~$i$,
for $i \in \{1,2\}$,
and these tokens are either
in the dispatcher station
or in the server-$i$ station.
The service policy
in the dispatcher station is processor-sharing,
while the service policies
in the server stations
match those applied by the servers.
The state of this network
is described by
the vector $x = (x_1, x_2)$
that counts the number of tokens
in the dispatcher station,
corresponding to available slots
in the servers' buffers. 
In particular, $x = 0$
means that both buffers are full,
while $x = \ell$,
with $\ell = (\ell_1, \ell_2)$,
means that the cluster is empty.
We let $e_1 = (1,0)$
and $e_2 = (0,1)$ denote
the states corresponding to
a single (slot occupied by a)
job at server~$i$.
Known results on closed Jackson
networks~\cite[Chapter~1 and Section~3.1]{S99}
show that the stationary distribution
of the Markov process defined
by the evolution of the network state is given by
\begin{align} \label{eq:stationary}
	\pi(x)
	= \frac1{G(\ell)}
	\binom{x_1 + x_2}{x_1}
	\left( \frac{\mu_1}\lambda \right)^{x_1}
	\left( \frac{\mu_2}\lambda \right)^{x_2},
	\quad x \le \ell,
\end{align}
where $G(\ell)$ is a normalization constant.

\section{Performance analysis} \label{sec:performance}

As in Jackson networks
with a single class of customers~\cite{B73},
the closed-form expressions
of several metrics of interest
stem directly from the normalization constant.

\subsection{Normalization constant} \label{subsec:performance-normalization}

The normalization constant follows
from the normalization equation
$\sum_{x \le \ell} \pi(x) = 1$.
Using~\eqref{eq:stationary}, we obtain
\begin{align} \label{eq:normalization}
	G(\ell)
	= \sum_{x \le \ell}
	\binom{x_1 + x_2}{x_1}
	\left( \frac{\mu_1}\lambda \right)^{x_1}
	\left( \frac{\mu_2}\lambda \right)^{x_2},
	\quad \ell \in \N^2.
\end{align}
The following geometric interpretation
guides the proofs of Theorems~\ref{theo:underloaded}
and~\ref{theo:overloaded}
in Sections~\ref{sec:underloaded}
and~\ref{sec:overloaded}.
Consider the square lattice
consisting of the 2-dimensional
vectors $x = (x_1, x_2)$
with integer components.
We are interested in the direct paths
going from the origin~0 to some vector~$x \le \ell$,
where by direct we mean
a path that consists only
of increasing unit steps
in the horizontal or vertical direction.
If each horizontal step has weight
$\frac{\mu_1}\lambda$
and each vertical step weight
$\frac{\mu_2}\lambda$,
then $(\frac{\mu_1}\lambda)^{x_1}
(\frac{\mu_2}\lambda)^{x_2}$
is the multiplicative weight
of any direct path going
from the origin~0 to the vector~$x$.
Since there are $\binom{x_1 + x_2}{x_1}$
such paths,
it follows that
the normalization constant $G(\ell)$
is the sum of the weights
of all direct paths
going from the origin
to a vector $x \le \ell$.
Intuitively, this suggests that:
\begin{itemize}
	\item If $\lambda \ll \mu_1 + \mu_2 (= 1)$,
	the paths that have the most weight
	are the longest,
	meaning that the states
	with large values of $x_1 + x_2$
	(that is, with many available slots)
	are the most likely.
	This intuition guides the proof
	of Theorem~\ref{theo:underloaded}.
	\item If $\lambda \gg 1$,
	the paths that have
	the most weight are the shortest,
	meaning that the states
	with small values of $x_1 + x_2$
	(that is, with few available slots)
	are the most likely.
	This intuition guides the proof
	of Theorem~\ref{theo:overloaded}.
\end{itemize}
Although~\eqref{eq:normalization}
unveils interesting properties
of the constant $G(\ell)$,
this expression is inconvenient
when it comes to \emph{compute} this constant
because it leads to numerical instability.
For the numerical results,
we use instead the recursive expression
\begin{align*}
	G(\ell)
	&= 1
	+ \sum_{\substack{i = 1 \\ \ell_i \ge 1}}^2
	\frac{\mu_i}\lambda G(\ell - e_i),
	\quad \ell \in \N^2 \setminus \{0\},
\end{align*}
with the base case $G(0) = 1$.
This expression,
which is a special case of the formula
derived in \cite[Proposition~2]{BJP04},
can be seen as a generalization
of the Erlang-B formula~\cite{E17}.

It will often
be convenient to consider the quantity
\begin{align} \label{eq:delta}
	\delta G(\ell)
	= G(\ell + e_1 - e_2) - G(\ell),
	\quad \ell \in \N^2: \ell_2 \ge 1,
\end{align}
that gives the variation
of the normalization constant
obtained by replacing
a server-2 slot with a server-1 slot.
By injecting \eqref{eq:normalization}
into this definition
and making simplifications,
we obtain
\begin{align}
	\label{eq:difference-n}
	\delta G(\ell)
	&= \sum_{n = \ell_1 + 1}^{\ell_1 + \ell_2}
	\binom{n}{\ell_1 + 1}
	\left( \frac{\mu_1}\lambda \right)^{\ell_1 + 1}
	\left( \frac{\mu_2}\lambda \right)^{n - \ell_1 - 1}
	- \sum_{n = \ell_2}^{\ell_1 + \ell_2}
	\binom{n}{\ell_2}
	\left( \frac{\mu_1}\lambda \right)^{n - \ell_2}
	\left( \frac{\mu_2}\lambda \right)^{\ell_2}.
\end{align}

\subsection{Long-term performance metrics} \label{subsec:performance-metrics}

We now consider three performance metrics
called the loss probability,
occupation rate,
and mean response time.
The formulas below
are simple extensions of formulas
derived for closed Jackson networks
with a single class of customers~\cite{B73}.
However, to the best of our knowledge,
the results regarding
the occupation rate and mean response time
have never been derived
in the literature
on insensitive load balancing.

\paragraph*{Loss probability}

The loss probability~$\beta(\ell)$
is defined as the probability
that an incoming job is rejected,
which happens when the buffers
of all servers are full upon its arrival.
According to the PASTA property~\cite{W82},
the loss probability is equal to
the stationary probability $\pi(0)$
that the buffers of all servers are full,
so that then by~\eqref{eq:stationary} we obtain
\begin{align} \label{eq:loss}
	\beta(\ell) = \frac1{G(\ell)}.
\end{align}
By~\eqref{eq:normalization},
the loss probability decreases
when the number of slots
in a given buffer increases
(as intuition suggests).

\paragraph*{Occupation rate}

For each $i \in \{1, 2\}$,
the occupation rate of server~$i$
is defined as the fraction of time
that this server is busy.
According to~\eqref{eq:stationary},
this quantity $\rho_i(\ell)$ is given by
\begin{align*}
	\frac1{G(\ell)}
	\sum_{x \le \ell - e_i}
	\binom{x_1 + x_2}{x_1}
	\left( \frac{\mu_1}\lambda \right)^{x_1}
	\left( \frac{\mu_2}\lambda \right)^{x_2}
	= \frac{G(\ell - e_i)}{G(\ell)}.
\end{align*}
We have $\rho_1(\ell) > \rho_2(\ell)$
if and only if
$G(\ell - e_1) > G(\ell - e_2)$,
which by~\eqref{eq:loss}
also means that the loss probability
increases less when we remove a slot
from server~1 than when
we remove a slot from server~2.

\paragraph*{Mean response time}

The response time of a job
is defined as the duration
between its arrival in the cluster
and its departure.
We can show that
the mean numbers of jobs
in the buffers of servers~1 and~2
are given by
\begin{align*}
	\alpha_1(\ell)
	&= \frac
	{\sum_{x_1 = 0}^{\ell_1 - 1} G(x_1, \ell_2)}
	{G(\ell)},
	&
	\alpha_2(\ell)
	&= \frac
	{\sum_{x_2 = 0}^{\ell_2 - 1} G(\ell_1, x_2)}
	{G(\ell)}.
\end{align*}
According to Little's law,
the mean response time of a job is given by
$
\Delta(\ell)
= \frac{\alpha_1(\ell) + \alpha_2(\ell)}
{\lambda (1 - \beta(\ell))}.
$

\subsection{Problem statement} \label{subsec:performance-problem}

We would like to
understand the joint impact of
the arrival rate~$\lambda$,
service rates~$\mu_1$ and~$\mu_2$,
and buffer lengths~$\ell_1$ and~$\ell_2$
on these performance metrics.
Despite the apparent simplicity of
these expressions,
using them to gain intuition
on the impact of parameters on performance
is a well-known difficult problem~\cite{H11}.
The results of Section~\ref{subsec:performance-metrics}
suggest however that the loss probability
is the easiest of these metrics to analyze,
as it is simply the inverse
of the normalization constant.

We observed earlier that,
although the load-balancing algorithm
does not account for the service rates
$\mu_1$ and $\mu_2$ to make the assignment decision,
the buffer lengths $\ell_1$ and $\ell_2$
can be adjusted to optimize performance.
Therefore, Sections~\ref{sec:underloaded}
to \ref{sec:monotonicity},
which contain our main contributions,
will focus more specifically
on the following question:
\begin{displayquote}
	Given the arrival rate~$\lambda$,
	service rates $\mu_1$ and $\mu_2$,
	and overall buffer length
	$L = \ell_1 + \ell_2$,
	which value(s) of $\ell_1$ and $\ell_2$
	minimize(s) the loss probability?
\end{displayquote}
Since $\mu_1 > \mu_2$,
a natural strategy consists of allocating
(almost) all slots to server~1.
However, the more jobs are
in service on server~1,
the smaller the fraction
of the service rate of
this server they receive,
and the longer they stay in the system,
thus preventing other jobs from entering.
This simple observation suggests
that it would be better to maximize the minimum
service rate received by any job
by allocating slots to servers
proportionally to their service rates,
that is, allocate a fraction $\mu_1$
of the slots to server~1
and a fraction $\mu_2$ of the slots
to server~2.
Theorem~\ref{theo:underloaded}
shows that this allocation
is indeed optimal
when the arrival rate is small.
Theorems~\ref{theo:overloaded}
and \ref{theo:monotonicity}
show that, as the arrival rate increases,
the optimal allocation
evolves monotonically towards
a uniform allocation,
in which both servers
have approximately the same number of slots.

Extending these results
to understand the impact
of parameters on the mean response time
is not straightforward.
This impact will be assessed
numerically in Section~\ref{sec:num}.

\section{Low-traffic regime}
\label{sec:underloaded}

Our first main contribution
is a theorem showing that,
when the arrival rate is small,
the loss probability is minimized
by allocating slots to servers
in proportion to their service rate.

\begin{theorem} \label{theo:underloaded}
	There is a $\lambda_* > 0$ such that,
	for each $\lambda \in (0, \lambda_*]$,
	the loss probability is minimized
	when $\ell = (\ell_1, L - \ell_1)$ with
	\begin{align} \label{eq:underloaded}
		\ell_1 =
		\left\lceil
		\mu_1 L - \mu_2
		\right\rceil
		\text{ or }
		\left\lfloor
		\mu_1 L + \mu_1
		\right\rfloor.
	\end{align}
\end{theorem}

\begin{proof}
	We first show that,
	as $\lambda \to 0$,
	the monotonicity of $G(\ell)$
	as a function of $\ell$
	is entirely dictated
	by that of the term
	corresponding to $x = \ell$
	in~\eqref{eq:normalization}.
	We will then study this term.
	
	Let $\ell = (\ell_1, \ell_2) \in \N^2$
	with $\ell_1 + \ell_2 = L$.
	First observe that
	\begin{align*}
		\sum_{\substack{
				x \le \ell: x \neq \ell
		}} \binom{x_1 + x_2}{x_1}
		\left( \frac{\mu_1}\lambda \right)^{x_1}
		\left( \frac{\mu_2}\lambda \right)^{x_2}
		&\le
		\frac{
			\left(
			\frac1\lambda
			\right)^L
			- 1
		}{
			\frac1\lambda
			- 1
		}.
	\end{align*}
	By injecting this inequality
	into~\eqref{eq:normalization}, we obtain
	\begin{align} \label{eq:binomial}
		G(\ell)
		&= \binom{L}{\ell_1}
		\left( \frac{\mu_1}\lambda \right)^{\ell_1}
		\left( \frac{\mu_2}\lambda \right)^{\ell_2}
		+ O_{\lambda \to 0} \left(
		\left(
		\frac1\lambda
		\right)^{L-1}
		\right).
	\end{align}
	If $\ell_2 \ge 1$,
	we can apply this result
	to both $\ell$ and $\ell + e_1 - e_2$,
	so that, by~\eqref{eq:delta}, we obtain
	\begin{align*}
		\delta G(\ell)
		&= \left(
		\frac{\ell_2}{\ell_1 + 1}
		\frac{\mu_1}{\mu_2} - 1
		\right)
		\cdot
		\binom{L}{\ell_1}
		\left( \frac{\mu_1}\lambda \right)^{\ell_1}
		\left( \frac{\mu_2}\lambda \right)^{\ell_2}
		+ O_{\lambda \to 0} \left(
		\left( \frac1\lambda \right)^{L-1}
		\right).
	\end{align*}
	As $\lambda$ tends to zero,
	the first term tends to $+\infty$
	like $(\frac1\lambda)^L$,
	while the second term tends to $+\infty$
	at most like $(\frac1\lambda)^{L-1}$.
	Therefore,
	there is $\lambda_* > 0$ such that,
	for each $\lambda \in (0, \lambda_*]$
	and each $\ell \in \N^2$ such that
	$\ell_1 + \ell_2 = L$
	and $\ell_2 \ge 1$,
	$\delta G(\ell)$
	is of the same sign as
	$(\frac{\ell_2}{\ell_1 + 1}
	\frac{\mu_1}{\mu_2} - 1)$,
	provided that this quantity is nonzero.
	
	Now let $\lambda \in (0, \lambda_*]$
	and assume that servers~$1$ and $2$
	have $\ell_1$ and $\ell_2$ slots,
	respectively, with $\ell_1 + \ell_2 = L$.
	According to the above equality,
	replacing a slot of server~$2$ (if any)
	with a slot of server~$1$
	reduces the loss probability
	whenever
	$\frac{\ell_2}{\ell_1 + 1}
	\frac{\mu_1}{\mu_2} < 1$, that is,
	$\ell_1
	< \mu_1 L
	- \mu_2$.
	By taking the symmetrical statement
	and rearranging the terms,
	we obtain that
	replacing a slot of server~$1$ (if any)
	with a slot of server~$2$
	reduces the loss probability
	whenever
	$\ell_1
	> \mu_1 L
	+ \mu_1$.
	Therefore, the slot allocation
	can only be optimal when
	$\mu_1 L - \mu_2
	\le \ell_1 \le
	\mu_1 L + \mu_1$,
	which is equivalent
	to~\eqref{eq:underloaded}.
\end{proof}

\begin{remark}
	The ceil and floor values
	in~\eqref{eq:underloaded}
	are different only in the pathological case
	where $\mu_1 L + \mu_1 = \mu_1 L - \mu_2 + 1$
	is an integer,
	which means that $\mu_1$ and $\mu_2$
	can be written as fractions
	with denominator $L + 1$.
	In this case, the term
	corresponding to $x = \ell$
	is zero in the expression of
	$\delta G(\ell)$,
	and which of
	$\ell_1
	= \left\lceil
	\mu_1 L - \mu_2
	\right\rceil$
	or
	$\ell_1
	= \left\lfloor
	\mu_1 L + \mu_1
	\right\rfloor$
	is optimal depends on the terms corresponding
	to $x \le \ell$ and $x \neq \ell$.
\end{remark}

\section{Heavy-traffic regime}
\label{sec:overloaded}

The following theorem shows that,
when the arrival rate $\lambda$ is large,
the loss probability is minimized
by taking $\ell_1 \simeq \ell_2$.
The arrival rates $\mu_1$ and $\mu_2$ only serve
to allocate the remaining slot
when the overall buffer length $L$ is odd.

\begin{theorem} \label{theo:overloaded}
	There is $\lambda^* > 0$ such that,
	for each $\lambda \in [\lambda^*, \infty)$,
	the loss probability is
	minimized by choosing:
	\begin{itemize}
		\item $\ell_1 = \ell_2 = \frac{L}2$
		if $L$ is even;
		\item $\ell_1 = \frac{L+1}2$
		and $\ell_2 = \frac{L-1}2$
		if $L$ is odd.
	\end{itemize}
\end{theorem}

\begin{proof}
	Let $\ell \in \N^2$ such that
	$\ell_1 + \ell_2 = L$
	and $\ell_1 \le L - 1$.
	Using~\eqref{eq:difference-n},
	we proceed by exhaustion,
	distinguishing several cases
	depending on the values
	of $\ell_1$ and $\ell_2$.
	
	\paragraph*{Case 1 ($\ell_1 \ge \ell_2$)}
	We can split the second sum
	of~\eqref{eq:difference-n}
	into two parts,
	corresponding to
	$n \in \{\ell_1 + 1, \ldots, \ell_1 + \ell_2\}$
	and $n \in \{\ell_2, \ldots, \ell_1\}$,
	respectively.
	We obtain
	\begin{align*}
		\delta G(\ell)
		&= \begin{aligned}[t]
			&\sum_{n = \ell_1 + 1}^{\ell_1 + \ell_2}
			\left[
				\binom{n}{\ell_1 + 1}
				\left( \frac{\mu_1}\lambda \right)
				^{\ell_1 + 1}
				\left( \frac{\mu_2}\lambda \right)
				^{n - \ell_1 - 1}
				- \binom{n}{\ell_2}
				\left( \frac{\mu_1}\lambda \right)
				^{n - \ell_2}
				\left( \frac{\mu_2}\lambda \right)
				^{\ell_2}
			\right] \\
			&- \sum_{n = \ell_2}^{\ell_1}
			\binom{n}{\ell_2}
			\left( \frac{\mu_1}\lambda \right)
			^{n - \ell_2}
			\left( \frac{\mu_2}\lambda \right)
			^{\ell_2}.
		\end{aligned}
	\end{align*}
	As $\lambda$ tends to $+\infty$,
	each term in the first sum
	tends to zero at least as fast
	as $(\frac1\lambda)^{\ell_1 + 1}$,
	while each term in the second sum
	tends to zero
	at most as fast as $(\frac1\lambda)^{\ell_1}$.
	Therefore, when $\lambda$ is sufficiently large,
	we have $\delta G(\ell) < 0$
	whenever $\ell_1 \ge \ell_2$.
	
	\paragraph*{Case 2 ($\ell_1 \le \ell_2 - 2$)}
	
	We can split the first sum
	of~\eqref{eq:difference-n}
	into two parts,
	corresponding to
	$n \in \{\ell_1 + 1, \ldots, \ell_2 - 1\}$
	and to $n \in
	\{\ell_2, \ldots, \ell_1 + \ell_2\}$,
	respectively.
	We obtain
	\begin{align*}
		\delta G(\ell)
		&= \begin{aligned}[t]
			&\sum_{n = \ell_1 + 1}
			^{\ell_2 - 1}
			\binom{n}{\ell_1 + 1}
			\left( \frac{\mu_1}\lambda \right)
			^{\ell_1 + 1}
			\left( \frac{\mu_2}\lambda \right)
			^{n - \ell_1 - 1} \\
			&+ \sum_{n = \ell_2}^{\ell_1 + \ell_2}
			\left[
				\binom{n}{\ell_1 + 1}
				\left( \frac{\mu_1}\lambda \right)
				^{\ell_1 + 1}
				\left( \frac{\mu_2}\lambda \right)
				^{n - \ell_1 - 1}
				- \binom{n}{\ell_2}
				\left( \frac{\mu_1}\lambda \right)
				^{n - \ell_2}
				\left( \frac{\mu_2}\lambda \right)
				^{\ell_2}
			\right].
		\end{aligned}
	\end{align*}
	As $\lambda$ tends to $+\infty$,
	each term in the first sum tends to zero
	at most as fast as $(\frac1\lambda)^{\ell_2 - 1}$,
	while each term in the second sum tends to zero
	at least as fast as $(\frac1\lambda)^{\ell_2}$.
	Therefore, when $\lambda$ is sufficiently large,
	we have $\delta G(\ell) > 0$
	whenever $\ell_1 \le \ell_2 - 2$.
	
	\paragraph*{Case 3 ($\ell_1 = \ell_2 - 1$,
		assuming that $L$ is odd)}
	
	The two sums in~\eqref{eq:difference-n}
	contain the same number of terms,
	and we obtain:
	\begin{align*}
		\delta G(\ell)
		&= \sum_{n = \ell_1 + 1}^{\ell_1 + \ell_2}
		\binom{n}{\ell_1 + 1}
			\left( \frac{\mu_1}\lambda \right)^
			{n - \ell_1 - 1}
			\left( \frac{\mu_2}\lambda \right)^
			{n - \ell_1 - 1}
			\times \left[
			\left( \frac{\mu_1}\lambda \right)
			^{\ell_1 + \ell_2 + 1 - n}
			- \left( \frac{\mu_2}\lambda \right)
			^{\ell_1 + \ell_2 + 1 - n}
			\right].
	\end{align*}
	Since $\mu_1 > \mu_2$, it follows that
	$\delta G(\ell) > 0$.
	
	\paragraph*{Conclusion}
	We now gather the three cases.
	If $L$ is even, the sequence
	$\ell_1 \mapsto \beta(\ell_1, L - \ell_1)$
	is decreasing on $\{0, 1, \ldots, \frac{L}2\}$
	and increasing on $\{\frac{L}2, \ldots, L-1, L\}$.
	Therefore, the loss probability is minimal
	when $\ell_1 = \ell_2 = \frac{L}2$.
	If $L$ is odd, the sequence
	$\ell_1 \mapsto \beta(\ell_1, L - \ell_1)$
	is decreasing on
	$\{0, 1, \ldots, \frac{L+1}2\}$
	and increasing on
	$\{\frac{L+1}2, \ldots, L-1, L\}$.
	Therefore, the loss probability is minimal
	when $\ell_1 = \frac{L+1}2$
	and $\ell_2 = \frac{L-1}2$.
\end{proof}

\section{Monotonicity} \label{sec:monotonicity}

To make the connection between
the results of the last two sections,
we now show that the optimal slot allocation
is monotonic with respect
to the arrival rate $\lambda$.
In the following theorem,
with ``optimal number of slots
allocated to the fastest server'',
we mean the \emph{smallest} number of slots
allocated to the fastest server
that minimizes the loss probability.
This will be discussed again in
Remark~\ref{remark:monotonicity}.

\begin{theorem} \label{theo:monotonicity}
	The optimal number of slots
	allocated to the fastest server,
	in terms of the loss probability,
	is decreasing with the arrival rate $\lambda$.
\end{theorem}

\begin{proof}
	It will be convenient to write
	the loss probability as
	a function $\beta(\lambda, \ell)$ of both
	the arrival rate $\lambda$
	and buffer lengths $\ell = (\ell_1, \ell_2)$.
	The proof relies on
	two monotonicity results
	that are presented
	in the following two propositions.
	
	\begin{proposition} \label{prop:monotonicity-ell}
		Let $\lambda > 0$
		and $\ell \in \N^2$
		such that $\ell_1 \ge 1$ and $\ell_2 \ge 1$.
		If $\beta(\lambda, \ell + e_1 - e_2)
		\le \beta(\lambda, \ell)$,
		then $\beta(\lambda, \ell - x e_ 1 + x e_2)
		< \beta(\lambda, \ell - (x+1) e_1 + (x+1) e_2)$
		for each $x \in \{0, 1, 2, \ldots, \ell_1 - 1\}$.
	\end{proposition}
	
	\begin{proposition} \label{prop:monotonicity-lambda}
		Let $\lambda_* > 0$
		and $\ell \in \N^2$
		such that $\ell_2 \ge 1$.
		If $\beta(\lambda_*, \ell + e_1 - e_2)
		< \beta(\lambda_*, \ell)$,
		then $\beta(\lambda, \ell + e_1 - e_2)
		< \beta(\lambda, \ell)$
		for each $\lambda \in (0, \lambda_*)$.
	\end{proposition}
	Propositions~\ref{prop:monotonicity-ell}
	and~\ref{prop:monotonicity-lambda}
	are proven in
	Appendices~\ref{app:monotonicity-ell}
	and~\ref{app:monotonicity-lambda},
	respectively.
	These propositions
	can be rephrased as follows.
	Assume that,
	for a given arrival rate $\lambda_*$
	and overall buffer length~$L$,
	we know that allocating $\ell_1 + 1$
	slots to the fastest server
	yields better performance
	than allocating $\ell_1$ slots
	to this server.
	Then Proposition~\ref{prop:monotonicity-ell}
	shows that, with the same
	arrival rate $\lambda_*$,
	allocating even fewer slots than $\ell_1$
	to the fastest server is even worse
	in terms of performance.
	Proposition~\ref{prop:monotonicity-lambda}
	shows that allocating $\ell_1+1$
	slots to the fastest server
	remains better than $\ell_1$ slots
	for any arrival rate $\lambda \in (0,\lambda_*)$.
	
	Now consider two arrival rates
	$\lambda_* \in (0, \infty)$
	and $\lambda \in (0, \lambda_*)$.
	Let $\ell^{\lambda_*}$ and $\ell^{\lambda}$
	denote the optimal slot allocations
	under these arrival rates,
	with $L = \ell_1^{\lambda_*} + \ell_2^{\lambda_*}
	= \ell_1^{\lambda} + \ell_2^{\lambda}$.
	Our objective is to prove that
	$\ell_1^{\lambda_*} \le \ell_1^{\lambda}$.
	The optimality of $\ell^{\lambda_*}$
	implies that
	$\beta(\lambda_*, \ell^{\lambda_*})
	< \beta(\lambda_*, \ell^{\lambda_*} - e_1 + e_2)$.
	Therefore, Proposition~\ref{prop:monotonicity-ell}
	gives
	$\beta(\lambda_*,
	\ell^{\lambda_*} - x e_1 + x e_2)
	< \beta(\lambda_*,
	\ell^{\lambda_*} - (x + 1) e_1 + (x + 1) e_2)$
	for each $x \in \{0, 1, 2, \ldots,
	\ell_1^{\lambda_*} - 1\}$.
	By Proposition~\ref{prop:monotonicity-lambda},
	each of these inequalities yields
	$\beta(\lambda,
	\ell^{\lambda_*} - x e_1 + x e_2)
	< \beta(\lambda,
	\ell^{\lambda_*} - (x + 1) e_1 + (x + 1) e_2)$
	for each $x \in \{0, 1, 2, \ldots,
	\ell_1^{\lambda_*} - 1\}$.
	This in turn implies that
	$\ell_1^{\lambda_*} \le \ell_1^{\lambda}$.
\end{proof}

\begin{remark} \label{remark:monotonicity}
	Proposition~\ref{prop:monotonicity-ell}
	shows that,
	among all vectors $\ell \in \N^2$
	such that $L = \ell_1 + \ell_2$,
	at most two can minimize
	the loss probability,
	and these are separated by only one slot.
	Therefore, in Theorem~\ref{theo:monotonicity},
	the ``optimal number of slots
	allocated to the fastest server''
	is defined up to plus or minus one,
	and we systematically choose the smallest
	value for convenience.
\end{remark}

\section{Extension to more than two servers}
\label{sec:generalization}

We now consider a cluster
that consists of
a dispatcher and
a set $\I = \{1, 2, \ldots, N\}$ of servers.
For each $i \in \I$, we let~$\mu_i$
denote the service rate of server~$i$,
$\ell_i$ the length of its buffer,
and $x_i$ the number of available slots
in this buffer.
We assume without loss of generality that
$1 > \mu_1 \ge \mu_2 \ge \ldots \ge \mu_N > 0$
and $\sum_{i \in \I} \mu_i = 1$,
and we let $L = \sum_{i \in \I} \ell_i$
denote the overall buffer length.
All definitions of
Sections~\ref{sec:model}
and \ref{sec:performance}
are generalized in a natural way
to this $N$-server cluster.
In particular,
the load-balancing algorithm is generalized as follows:
an incoming job is assigned to server~$i$ with
probability $\frac{x_i}{\sum_{j \in \I} x_j}$
for each $i \in \I$
if $\sum_{j \in \I} x_j \ge 1$,
otherwise the job is rejected.
The corresponding
closed Jackson network
consists of $N+1$ stations
that correspond to the dispatcher
and the $N$ servers, respectively.
The set of token classes is~$\I$
and there are $\ell_i$ class-$i$ tokens,
for each $i \in \I$.

\paragraph*{Stationary distribution}

The stationary distribution of
the Markov process
defined by the evolution
of the network state
$x = (x_1, x_2, \ldots, x_N)$
over time is given by
\begin{align*}
	\pi(x)
	= \frac1{G(\ell)}
	\binom{x_1 + x_2 + \ldots + x_N}
	{x_1, x_2, \ldots, x_N}
	\prod_{i = 1}^N
	\left( \frac{\mu_i}\lambda \right)^{x_i},
	\quad x \le \ell,
\end{align*}
where
$\binom{x_1 + x_2 + \ldots + x_N}
{x_1, x_2, \ldots, x_N}
= \frac{(x_1 + x_2 + \ldots + x_N)!}
{x_1! x_2! \cdots x_N!}$
is a multinomial coefficient,
and the constant~$G(\ell)$
is obtained by normalization.
For each $i, j \in \I$
and $\ell \in \N^N$
with $\ell_j \ge 1$,
the variation of the normalization constant
obtained by replacing a server-$j$ slot
with a server-$i$ slot is denoted by
$\delta_{j \to i} G(\ell)
= G(\ell + e_i - e_j) - G(\ell)$.
We now explain how to generalize
our results.

\paragraph*{Theorem~\ref{theo:underloaded}}

Following the same approach as
in the proof of Theorem~\ref{theo:underloaded},
we obtain the following
generalization of~\eqref{eq:binomial}:
\begin{align*}
	G(\ell)
	&=
	\binom{L}{\ell_1, \ell_2, \ldots, \ell_N}
	\prod_{i = 1}^N
	\left( \frac{\mu_i}\lambda \right)^{\ell_i}
	+ O_{\lambda \to 0} \left(
	\left( \frac1\lambda \right)^{L-1}
	\right).
\end{align*}
As $\lambda$ tends to zero,
the variations of the second term
become negligible compared to
those of the first.
Therefore,
an optimal slot allocation
is a mode
of the multinomial distribution
with parameters $L$
and $\mu_1, \mu_2, \ldots, \mu_N$.

\paragraph*{Theorem~\ref{theo:overloaded}}

To generalize the proof
of this theorem and the next,
it is helpful to rewrite
the stationary distribution as
\begin{align} \label{eq:stationary-N}
	\pi(x)
	= \frac1{G(\ell)}
	\varphi_1(x_1, x_2)
	\varphi_2(x_1 + x_2, x_{-1,2}),
	\quad x \le \ell,
\end{align}
where
$x_{-1,2} = (x_3, \ldots, x_N)$,
and
\begin{align*}
	\varphi_1(x_1, x_2)
	&= \binom{x_1 + x_2}{x_1}
	\left( \frac{\mu_1}\lambda \right)^{x_1}
	\left( \frac{\mu_2}\lambda \right)^{x_2},
	&
	\varphi_2(x_1 + x_2, x_{-1,2})
	&= \binom{x_1 + x_2 + \ldots + x_N}
	{x_1 + x_2, x_3, \ldots, x_N}
	\prod_{i = 3}^N
	\left( \frac{\mu_i}\lambda \right)^{x_i}.
\end{align*}
The first factor
is equal (up to a multiplicative constant)
to the stationary distribution
obtained in a two-server cluster,
while the second factor
depends on $x_1$ and $x_2$
only via their sum $x_1 + x_2$.
Using this expression,
we can rewrite $G(\ell)$,
and then $\delta_{2 \to 1} G(\ell)$,
as a nested sum
over $x_1 \in \{0, 1, \ldots, \ell_1\}$
and $x_2 \in \{0, 1, \ldots, \ell_2\}$,
where the term corresponding
to $x_1$ and $x_2$ is the product
of $\varphi_1(x_1, x_2)$
and a factor that depends
on $x_1$ and $x_2$ only via their sum.
Upon observing that this factor is
$\Theta_{\lambda \to 0}
((\frac1\lambda)^{\ell_3 + \ldots + \ell_N})$,
we conclude in a similar way
as in the proof of
Theorem~\ref{theo:underloaded}
that, when $\lambda$ is small enough,
we have $\delta_{2 \to 1} G(\ell) < 0$
if $\ell_1 \ge \ell_2$
and  $\delta_{2 \to 1} G(\ell) > 0$
if $\ell_1 \le \ell_2 - 1$.
A similar result holds
for $G_{i \to j}(\ell)$
for each $i, j \in \I$ such that $\ell_j \ge 1$,
so that the loss probability
is again minimized
by choosing the buffer lengths
approximately equal to each other.

\paragraph*{Theorem~\ref{theo:monotonicity}}

The monotonicity result of this theorem
can only be generalized to
the fastest and slowest servers.
More specifically,
we can show that
the optimal buffer length
of the fastest server(s)
decreases with the arrival rate,
while the optimal buffer length
of the slowest server(s)
increases with the arrival rate.
Indeed, all intermediary results
in Appendices~\ref{app:monotonicity-ell}
and \ref{app:monotonicity-lambda}
can be generalized to two arbitrary
servers~$i, j \in \I$
using the same approach as
in the previous paragraph.
The only restriction is that
$\mu_i \ge \mu_j$
(the assumption $\mu_i > \mu_j$ is used
in Case~1 in the proof of
Lemma~\ref{lem:monotonicity-lambda}
in Appendix~\ref{app:monotonicity-lambda},
and one can verify that
it can be alleviated to
$\mu_i \ge \mu_j$).
Therefore, we can only conclude
for a server $i \in \I$ such that
either $\mu_i \ge \mu_j$
for each $j \in \I \setminus \{i\}$
or $\mu_i \le \mu_j$
for each $j \in \I \setminus \{i\}$.

\section{Numerical results} \label{sec:num}

We now proceed to some case studies, in which we numerically evaluate the performance measures from Section~\ref{subsec:performance-metrics} and show how the loss probability and mean response time are dependent on the arrival rate, service times and slot allocation.

\subsection{Cluster of two servers}

We first consider a cluster of two servers.
The overall buffer length is kept constant equal to $L=20$ in the interest of space, but we observed a similar behavior for other values of~$L$. As before, we assume that $\mu_1 + \mu_2 = 1$ and $\mu_1 > \mu_2$.

\paragraph*{Loss probability}

In \figurename~\ref{fig:nb-tokens-vs-loss},
the loss probability is shown
as a function of the number of slots allocated to the fast server, for several values of the arrival rate $\lambda$,
first in a linear plot
and then in a log-linear plot.
As intuition suggests, a larger arrival rate yields a larger loss probability.
Since $\mu_1$ is large,
the loss probability tends to be lower when the first server has more than half of the slots.
The log-linear plot
reveals that, even in
this range, the loss probability
can still be reduced by
several orders of magnitude
by correctly choosing the buffer lengths.

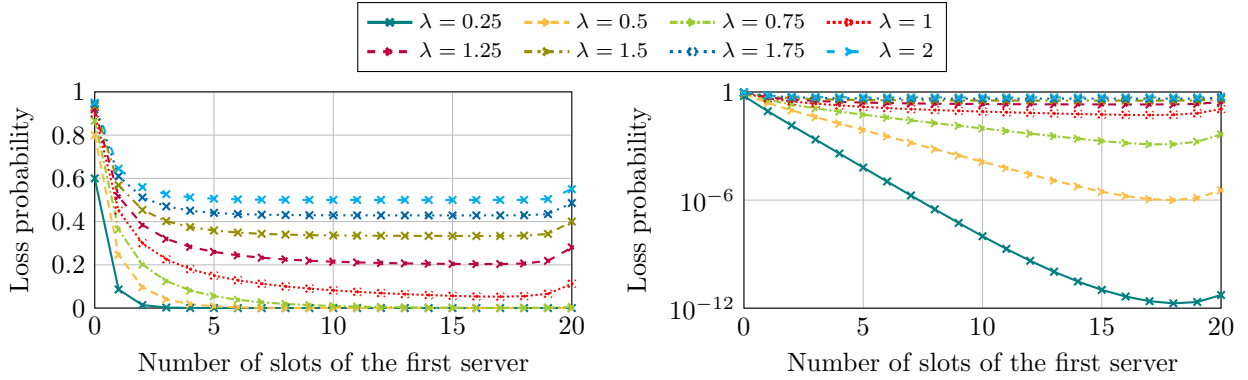
\begin{figure}[ht]
	\centering
	\hskip .8cm
	\begin{tikzpicture}
		\begin{axis}[legendplotstyle]
			\addlegendimage{teal, mark=x}
			\addlegendentry{$\lambda = 0.25$};
			
			\addlegendimage{orange, densely dashed, mark=x}
			\addlegendentry{$\lambda = 0.5$};
			
			\addlegendimage{green, densely dashdotted, mark=x}
			\addlegendentry{$\lambda = 0.75$};
			
			\addlegendimage{red, densely dotted, mark=x,}
			\addlegendentry{$\lambda = 1$};
			
			\addlegendimage{purple, dashed, mark=x}
			\addlegendentry{$\lambda = 1.25$};
			
			\addlegendimage{olive, dashdotted, mark=x}
			\addlegendentry{$\lambda = 1.5$};
			
			\addlegendimage{blue, dotted, mark=x}
			\addlegendentry{$\lambda = 1.75$};
			
			\addlegendimage{cyan, loosely dashed, mark=x}
			\addlegendentry{$\lambda = 2$};
		\end{axis}
	\end{tikzpicture}
	\\
	\raggedleft
	\pgfplotstableread{figure3.csv}\mytable
	\begin{tikzpicture}
		\begin{axis}[defaultplotstyle,
			xlabel={Number of slots of the first server},
			ylabel={Loss probability},
			xmin=0, xmax=20,
			ymin=0, ymax=1,
			width=.48\linewidth,
			]
			
			\addplot+[
			teal, mark=x,
			] table[x={"l1"}, y={"lambda=0.25"}]{\mytable};
			
			\addplot+[
			orange, densely dashed, mark=x,
			] table[x={"l1"}, y={"lambda=0.5"}]{\mytable};
			
			\addplot+[
			green, densely dashdotted, mark=x,
			] table[x={"l1"}, y={"lambda=0.75"}]{\mytable};
			
			\addplot+[
			red, densely dotted, mark=x,
			] table[x={"l1"}, y={"lambda=1."}]{\mytable};
			
			\addplot+[
			purple, dashed, mark=x,
			] table[x={"l1"}, y={"lambda=1.25"}]{\mytable};
			
			\addplot+[
			olive, dashdotted, mark=x,
			] table[x={"l1"}, y={"lambda=1.5"}]{\mytable};
			
			\addplot+[
			blue, dotted, mark=x,
			] table[x={"l1"}, y={"lambda=1.75"}]{\mytable};
			
			\addplot+[
			cyan, loosely dashed, mark=x,
			] table[x={"l1"}, y={"lambda=2."}]{\mytable};
			
		\end{axis}
	\end{tikzpicture}
	\hfill
	\pgfplotstableread{figure3.csv}\mytable
	\begin{tikzpicture}
		\begin{semilogyaxis}[defaultplotstyle,
			xlabel={Number of slots of the first server},
			ylabel={Loss probability},
			yticklabels={$10^{-12}$, $10^{-6}$, $1$},
			xmin=0, xmax=20,
			ymin=0.000000000001, ymax=1,
			width=.48\linewidth,
			]
			
			\addplot+[
			teal, mark=x,
			] table[x={"l1"},
			y={"lambda=0.25"}]{\mytable};
			
			\addplot+[
			orange, densely dashed, mark=x,
			] table[x={"l1"},
			y={"lambda=0.5"}]{\mytable};
			
			\addplot+[
			green, densely dashdotted, mark=x,
			] table[x={"l1"}, 	y={"lambda=0.75"}]{\mytable};
			
			\addplot+[
			red, densely dotted, mark=x,
			] table[x={"l1"}, 	y={"lambda=1."}]{\mytable};
			
			\addplot+[
			purple, dashed, mark=x,
			] table[x={"l1"}, 	y={"lambda=1.25"}]{\mytable};
			
			\addplot+[
			olive, dashdotted, mark=x,
			] table[x={"l1"}, y={"lambda=1.5"}]{\mytable};
			
			\addplot+[
			blue, dotted, mark=x,
			] table[x={"l1"}, y={"lambda=1.75"}]{\mytable};
			
			\addplot+[
			cyan, loosely dashed, mark=x,
			] table[x={"l1"}, y={"lambda=2."}]{\mytable};
		\end{semilogyaxis}
	\end{tikzpicture}
	\caption{Loss probability
		in a two-server cluster,
		as a function of
		the buffer length of the first server
		and for several values of the arrival rate,
		with $L = 20$, $\mu_1=0.9$,
		and $\mu_2=0.1$.}
	\label{fig:nb-tokens-vs-loss}
\end{figure}

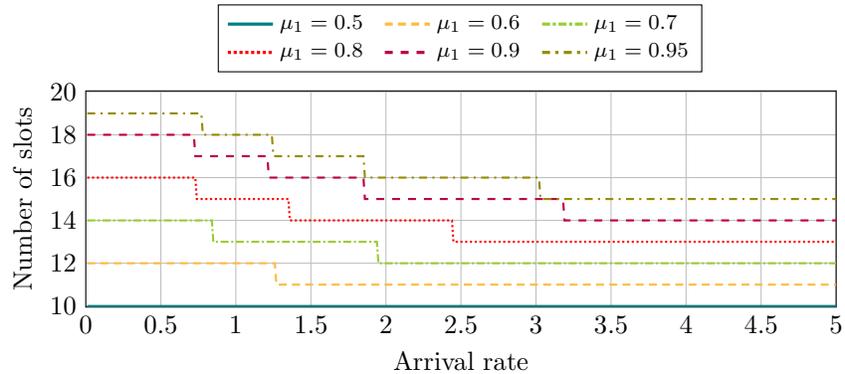
\begin{figure}[ht]
	\centering
	\hskip .85cm
	\begin{tikzpicture}
		\begin{axis}[legendplotstyle,
			legend columns=3]
			\addlegendimage{teal, no markers}
			\addlegendentry{$\mu_1 = 0.5$};
			
			\addlegendimage{orange, densely dashed, no markers}
			\addlegendentry{$\mu_1 = 0.6$};
			
			\addlegendimage{green, densely dashdotted, no markers}
			\addlegendentry{$\mu_1 = 0.7$};
			
			\addlegendimage{red, densely dotted, no markers,}
			\addlegendentry{$\mu_1 = 0.8$};
			
			\addlegendimage{purple, dashed, no markers}
			\addlegendentry{$\mu_1 = 0.9$};
			
			\addlegendimage{olive, dashdotted, no markers}
			\addlegendentry{$\mu_1 = 0.95$};
		\end{axis}
	\end{tikzpicture}
	\\
	\pgfplotstableread{figure4.csv}\mytable
	\begin{tikzpicture}
		\begin{axis}[defaultplotstyle,
			xlabel={Arrival rate},
			ylabel={Number of slots},
			xmin=0, xmax=5,
			ymin=9.95, ymax=20.05,
			]
			
			\addplot+[
			teal, no markers,
			] table[x={"lambda"}, y={"mu1=0.5"}]{\mytable};
			
			\addplot+[
			orange, densely dashed, no markers,
			] table[x={"lambda"}, y={"mu1=0.6"}]{\mytable};
			
			\addplot+[
			green, densely dashdotted, no markers,
			] table[x={"lambda"}, y={"mu1=0.7"}]{\mytable};
			
			\addplot+[
			red, densely dotted, no markers,
			] table[x={"lambda"}, y={"mu1=0.8"}]{\mytable};
			
			\addplot+[
			purple, dashed, no markers,
			] table[x={"lambda"}, y={"mu1=0.9"}]{\mytable};
			
			\addplot+[
			olive, dashdotted, no markers,
			] table[x={"lambda"}, y={"mu1=0.95"}]{\mytable};
			
		\end{axis}
	\end{tikzpicture}
	\caption{Optimal buffer length of
		the fast server
		(to minimize the loss probability)
		in a two-server cluster,
		as a function of the arrival rate and
		for several values of the service rates,
		with $L=20$.}
	\label{fig:lambda-vs-optimal-loss}
\end{figure}

To make more decisive statements, we consider, in \figurename~\ref{fig:lambda-vs-optimal-loss}, the optimal buffer length of the fast server (called the \emph{optimal buffer length} for brevity) as a function of the arrival rate, for several values of the service rates. The optimal buffer length is always at least $\lceil L/2 \rceil$, as it does not make sense to allocate more slots to a slower server, and it increases with $\mu_1$.

When the arrival rate $\lambda$ is small, the optimal buffer length is approximately $\mu_1 L$, which is consistent with Theorem~\ref{theo:underloaded}. Although it may at first \emph{seem} that the larger the difference in server speeds, the later the optimal number of slots changes when increasing $\lambda$, this is a coincidence, as can be seen by comparing the cases $\mu_1=0.9$ and $\mu_1=0.95$. As $\lambda$ increases, the optimal buffer length converges to $L/2 = 10$, as predicted by Theorem~\ref{theo:overloaded}, but the convergence is slower when $\mu_1$ is larger. \figurename~\ref{fig:lambda-vs-optimal-loss} lastly shows that the optimal buffer length decreases as the arrival rate increases, as proven in Theorem~\ref{theo:monotonicity}.

\paragraph*{Mean response time}

We now turn to the mean response time, for which we only show numerical results. The mean response time is shown in \figurename~\ref{fig:arrival-rate-vs-mst}
as a function of the arrival rate $\lambda$, for several values of the buffer length of the first server.
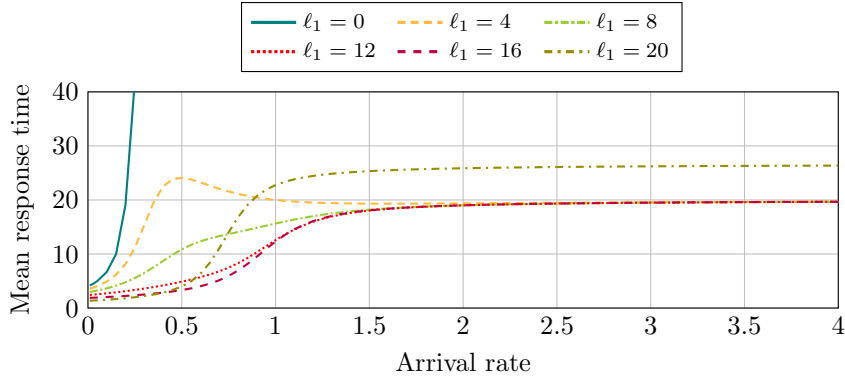
\begin{figure}[t]
	\centering
	\hskip .9cm
	\begin{tikzpicture}
		\begin{axis}[legendplotstyle,
			legend columns=3]
			\addlegendimage{teal, no markers}
			\addlegendentry{$\ell_1 = 0$};
			
			\addlegendimage{orange, densely dashed, no markers}
			\addlegendentry{$\ell_1 = 4$};
			
			\addlegendimage{green, densely dashdotted, no markers}
			\addlegendentry{$\ell_1 = 8$};
			
			\addlegendimage{red, densely dotted, no markers,}
			\addlegendentry{$\ell_1 = 12$};
			
			\addlegendimage{purple, dashed, no markers}
			\addlegendentry{$\ell_1 = 16$};
			
			\addlegendimage{olive, dashdotted, no markers}
			\addlegendentry{$\ell_1 = 20$};=
		\end{axis}
	\end{tikzpicture}
	\\
	\pgfplotstableread{figure5.csv}\mytable
	\begin{tikzpicture}
		\begin{axis}[defaultplotstyle,
			xlabel={Arrival rate},
			ylabel={Mean response time},
			xmin=0, xmax=4,
			ymin=0, ymax=40,
			]
			
			\addplot+[
			teal, no markers,
			] table[x={"lambda"}, y={"l1=0"}]{\mytable};
			
			\addplot+[
			orange, densely dashed, no markers,
			] table[x={"lambda"}, y={"l1=4"}]{\mytable};
			
			\addplot+[
			green, densely dashdotted, no markers,
			] table[x={"lambda"}, y={"l1=8"}]{\mytable};
			
			\addplot+[
			red, densely dotted, no markers,
			] table[x={"lambda"}, y={"l1=12"}]{\mytable};
			
			\addplot+[
			purple, dashed, no markers,
			] table[x={"lambda"}, y={"l1=16"}]{\mytable};
			
			\addplot+[
			olive, dashdotted, no markers,
			] table[x={"lambda"}, y={"l1=20"}]{\mytable};
		\end{axis}
	\end{tikzpicture}
	\caption{Mean response time
		in a two-server cluster,
		as a function of the arrival rate and
		for several values of the buffer length
		of the first server, with
		$L = 20$ and $\mu_1 = 0.75$.
	}
	\label{fig:arrival-rate-vs-mst}
\end{figure}%
We first turn to the paradoxical behavior of the system when the majority of slots is allocated to the slowest server. In this case, the mean response time is \emph{not} necessarily monotonous in the arrival rate $\lambda$, as can be seen with $\ell_1 = 4$.
Indeed, when the arrival rate is low, most jobs are served by the slowest server; as the arrival rate increases, a larger fraction of jobs is sent to the fastest server (even if the buffer of this server is shorter), so that the mean response time decreases.
Note that this scenario is suboptimal anyway, as switching around the slots (i.e.\ giving the majority of slots to the fast server) leads to better performance.
As a side remark, the mean response times always converges to the same value as $\lambda$ increases, as long as at least one slot is allocated to each server.

\figurename~\ref{fig:lambda-vs-optimal-mst}
shows the optimal buffer length of the fast server (to minimize the mean response time) as a function of the arrival rate $\lambda$. The same monotonicity property as for the loss probability seems to hold. Furthermore, when the arrival rate is large, the optimal buffer length also seems converge to $L/2$.
The main difference with the loss probability appears in the low-traffic regime: instead of converging to $\lceil \mu_1 L \rceil$, the optimal buffer length of the fastest server seems to converge to~$L$. 

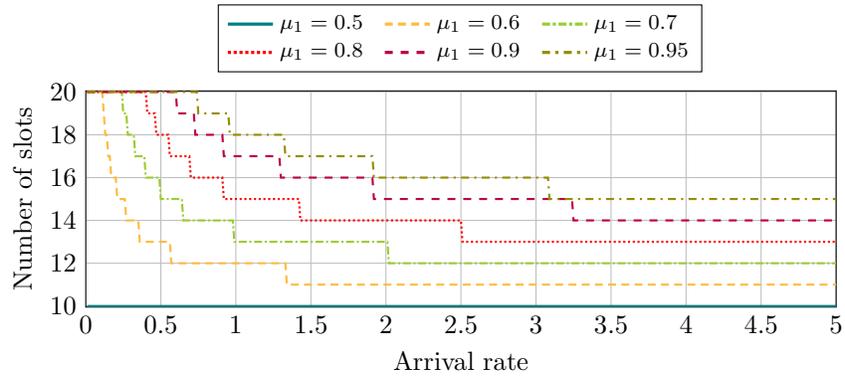
\begin{figure}[t]
	\centering
	\hskip .85cm
	\begin{tikzpicture}
		\begin{axis}[legendplotstyle,
			legend columns=3]
			\addlegendimage{teal, no markers}
			\addlegendentry{$\mu_1 = 0.5$};
			
			\addlegendimage{orange, densely dashed, no markers}
			\addlegendentry{$\mu_1 = 0.6$};
			
			\addlegendimage{green, densely dashdotted, no markers}
			\addlegendentry{$\mu_1 = 0.7$};
			
			\addlegendimage{red, densely dotted, no markers,}
			\addlegendentry{$\mu_1 = 0.8$};
			
			\addlegendimage{purple, dashed, no markers}
			\addlegendentry{$\mu_1 = 0.9$};
			
			\addlegendimage{olive, dashdotted, no markers}
			\addlegendentry{$\mu_1 = 0.95$};
		\end{axis}
	\end{tikzpicture}
	\\
	\pgfplotstableread{figure6.csv}\mytable
	\begin{tikzpicture}
		\begin{axis}[defaultplotstyle,
			xlabel={Arrival rate},
			ylabel={Number of slots},
			xmin=0, xmax=5,
			ymin=9.95, ymax=20.05,
			]
			
			\addplot+[
			teal, no markers,
			] table[x={"lambda"}, y={"mu1=0.5"}]{\mytable};
			
			\addplot+[
			orange, densely dashed, no markers,
			] table[x={"lambda"}, y={"mu1=0.6"}]{\mytable};
			
			\addplot+[
			green, densely dashdotted, no markers,
			] table[x={"lambda"}, y={"mu1=0.7"}]{\mytable};
			
			\addplot+[
			red, densely dotted, no markers,
			] table[x={"lambda"}, y={"mu1=0.8"}]{\mytable};
			
			\addplot+[
			purple, dashed, no markers,
			] table[x={"lambda"}, y={"mu1=0.9"}]{\mytable};
			
			\addplot+[
			olive, dashdotted, no markers,
			] table[x={"lambda"}, y={"mu1=0.95"}]{\mytable};
			
		\end{axis}
	\end{tikzpicture}
	\caption{Optimal buffer length of the fast server
		in a two-server cluster
		(to minimize
		the mean response time),
		as a function of the arrival rate and
		for several values of the service rates,
		with $L=20$.}
	\label{fig:lambda-vs-optimal-mst}
\end{figure}

\subsection{Cluster of four servers}

We finally consider a cluster
with four servers to illustrate
the absence of monotonicity
of the optimal buffer length
in larger clusters.
The server speeds are
$\mu_1 = 0.45$, $\mu_2=0.3$,
$\mu_3=0.2$, and $\mu_4=0.05$
and overall buffer size $L = 40$.
The optimal buffer lengths for each server,
in terms of either the loss probability
or the mean response time,
are shown in
\figurename ~\ref{fig:lambda-vs-optimal-loss-large}
and \ref{fig:lambda-vs-optimal-time-large}
as functions of the arrival rate.

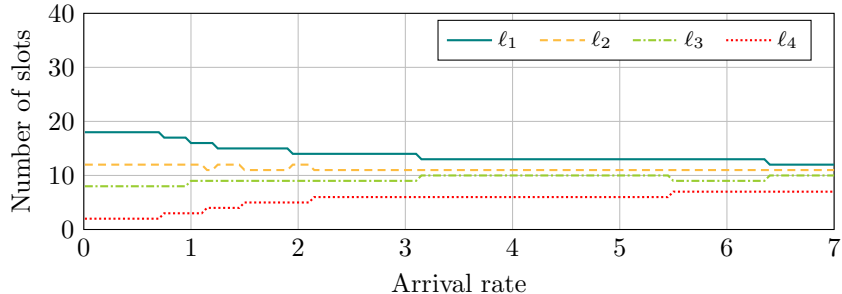
\begin{figure}[t]
	\centering
	\pgfplotstableread{figure7.csv}\mytable
	\begin{tikzpicture}
		\begin{axis}[defaultplotstyle,
			xlabel={Arrival rate},
			ylabel={Number of slots},
			xmin=0, xmax=7,
			ymin=0, ymax=40,
			legend columns=4,
			legend pos=north east,
			xtick={0,1,2,3,4,5,6,7},
			]
			
			\addplot+[
			teal, no markers,
			] table[x={"lambda"}, y={"l1"}]{\mytable};
			\addlegendentry{$\ell_1$};
			
			\addplot+[
			orange, densely dashed, no markers,
			] table[x={"lambda"}, y={"l2"}]{\mytable};
			\addlegendentry{$\ell_2$};
			
			\addplot+[
			green, densely dashdotted, no markers,
			] table[x={"lambda"}, y={"l3"}]{\mytable};
			\addlegendentry{$\ell_3$};
			
			\addplot+[
			red, densely dotted, no markers,
			] table[x={"lambda"}, y={"l4"}]{\mytable};
			\addlegendentry{$\ell_4$};
			
		\end{axis}
	\end{tikzpicture}
	\caption{Optimal buffer length of each server
		(to minimize the loss probability)
		in a four-server cluster,
		as function of $\lambda$,
		with $L = 40$,
		$\mu_1 = 0.45$, $\mu_2=0.3$, $\mu_3=0.2$, $\mu_4=0.05$.}
	\label{fig:lambda-vs-optimal-loss-large}
\end{figure}

As mentioned in Section~\ref{sec:generalization}, all theorems for the loss probability can be generalized to more than two servers: the optimal buffer lengths are proportional to the server speeds when the arrival rate is low and uniform when the arrival rate is large, and the optimal buffer length of the fastest and slowest servers evolve monotonically with the arrival rate. The optimal buffer lengths of servers~2 and~3 are however not monotonic. From extensive numerical experiments conducted for clusters with four to ten servers, we conjecture that, for clusters with $N$ servers, the total optimal buffer size of the $n=1,2,\ldots,N$ fastest servers is decreasing, e.g., with $N = 4$, $\ell_1$, $\ell_1+\ell_2$, and $\ell_1+\ell_2+\ell_3$ are decreasing in $\lambda$.

We observe in \figurename~\ref{fig:lambda-vs-optimal-time-large} that when the arrival rate is low, the optimal buffer size of the fastest server (in terms of mean response time) is $L$, as in the two-server case. For a large arrival rate, the optimal buffer lengths are uniform. A similar sense of monotonicity is observed for the mean response time: the optimal buffer lengths of the fastest and slowest servers are monotonous. The total optimal buffer length of the $n=1,2,\hdots,L$ fastest servers also seems to be monotonous.

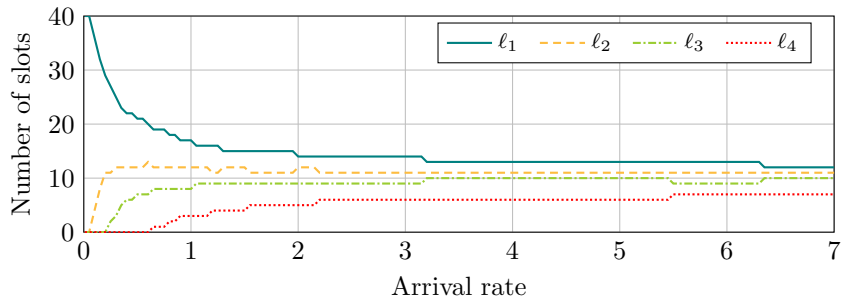
\begin{figure}[t]
	\centering
	\pgfplotstableread{figure8.csv}\mytable
	\begin{tikzpicture}
		\begin{axis}[defaultplotstyle,
			xlabel={Arrival rate},
			ylabel={Number of slots},
			xmin=0, xmax=7,
			ymin=0, ymax=40,
			xtick={0,1,2,3,4,5,6,7},
			legend pos=north east,
			legend columns=4,
			]
			
			\addplot+[
			teal, no markers,
			] table[x={"lambda"}, y={"l1"}]{\mytable};
			\addlegendentry{$\ell_1$};
			
			\addplot+[
			orange, densely dashed, no markers,
			] table[x={"lambda"}, y={"l2"}]{\mytable};
			\addlegendentry{$\ell_2$};
			
			\addplot+[
			green, densely dashdotted, no markers,
			] table[x={"lambda"}, y={"l3"}]{\mytable};
			\addlegendentry{$\ell_3$};
			
			\addplot+[
			red, densely dotted, no markers,
			] table[x={"lambda"}, y={"l4"}]{\mytable};
			\addlegendentry{$\ell_4$};
			
		\end{axis}
	\end{tikzpicture}
	\caption{Optimal buffer length of each server
		(to minimize the mean response time)
		in the same setting as in \figurename~\ref{fig:lambda-vs-optimal-loss-large}.}
	\label{fig:lambda-vs-optimal-time-large}
\end{figure}

\section{Conclusion} \label{sec:ccl}

In this paper, we considered
a load-balancing algorithm
that leads to a product-form
stationary distribution
in clusters of servers
with unequal service speeds.
We developed analytical methods
to understand the joint impact of
the speeds and buffer lengths
of the servers on performance.
For a two-server cluster
with arbitrary service speeds,
we proved analytically that
the optimal slot allocation
in terms of the loss probability
evolves monotonically
from proportional to the server speeds to uniform
as the arrival rate increases.
We then generalized these results
to clusters of more than two servers
and assessed their validity
on numerical examples.

For the future works, we would like to
generalize these analytical results
to the other performance metrics
mentioned in the introduction.
The main difficulty is that their expressions
involve fractions of two sums
that both span all states.
We would also like
to generalize these results
in other directions,
for instance by considering
open variants of the algorithm
or by accounting for assignment constraints~\cite{C19-1,WZS20}.

\section*{Acknowledgments}

The research of Mark van der Boor is supported by
the NWO Gravitation Networks grant 024.002.003.
The authors thank Sem Borst
for his valuable comments
on an earlier draft of the paper,
and in particular for suggesting
the generalization of Theorem~\ref{theo:underloaded}
to clusters with more than two servers.
The authors are also grateful
to Mor Harchol-Balter
for an insightful discussion
on Theorems~\ref{theo:underloaded}
and \ref{theo:overloaded}.


\appendix

\section*{Appendix}

In the following two proofs,
the arrival $\lambda$
will be mentioned explicitly
as a parameter of the quantities involved,
so that in particular we will let
$G(\lambda, \ell)$
denote the normalization constant
under arrival rate~$\lambda$
and buffer length vector~$\ell$,
and $\delta G(\lambda, \ell)
= G(\lambda, \ell + e_1 - e_2) - G(\lambda, \ell)$.
Also, it will be convenient
to rewrite the normalization constant
$G(\lambda, \ell)$ as
\begin{align} \label{eq:betabarpi}
	G(\lambda, \ell)
	= \sum_{x \le \ell}
	\bar\pi(x),
	\quad \ell \in \N^2,
\end{align}
where $\bar\pi$ is a stationary measure
such that $\bar\pi(0) = 1$, that is,
\begin{align} \label{eq:barpi}
	\bar\pi(x) = \binom{x_1 + x_2}{x_1}
	\left( \frac{\mu_1}\lambda \right)^{x_1}
	\left( \frac{\mu_2}\lambda \right)^{x_2},
	\quad x \in \N^2.
\end{align}
Injecting this definition
into the definition of $\delta G(\ell)$ yields
\begin{align} \label{eq:deltaG}
	\delta G(\lambda, \ell)
	&= \sum_{x_2 = 0}^{\ell_2 - 1}
	\bar\pi(\ell_1 + 1, x_2)
	- \sum_{x_1 = 0}^{\ell_1}
	\bar\pi(x_1, \ell_2).
\end{align}
Appendices~\ref{app:monotonicity-ell}
and~\ref{app:monotonicity-lambda}
give the proofs of
Propositions~\ref{prop:monotonicity-ell}
and~\ref{prop:monotonicity-lambda},
respectively.

\section{Proof of Proposition~\ref{prop:monotonicity-ell}}
\label{app:monotonicity-ell}

\theoremstyle{plain}
\newtheorem*{prop:monotonicity-ell}
{Proposition~\ref{prop:monotonicity-ell}}
\begin{prop:monotonicity-ell}
	Let $\lambda > 0$
	and $\ell \in \N^2$
	such that $\ell_1 \ge 1$ and $\ell_2 \ge 1$.
	If $\beta(\lambda, \ell + e_1 - e_2)
	\le \beta(\lambda, \ell)$,
	then $\beta(\lambda, \ell - x e_ 1 + x e_2)
	< \beta(\lambda, \ell - (x+1) e_1 + (x+1) e_2)$
	for each $x \in \{0, 1, 2, \ldots, \ell_1 - 1\}$.
\end{prop:monotonicity-ell}

Consider an arrival rate $\lambda > 0$.
We will prove that,
for each $\ell \in \N^2$
such that $\ell_1 \ge 1$ and $\ell_2 \ge 1$,
$\beta(\lambda, \ell + e_1 - e_2)
\le \beta(\lambda, \ell)$
implies $\beta(\lambda, \ell)
< \beta(\lambda, \ell - e_1 + e_2)$,
that is,
$\delta G(\lambda, \ell) > 0$
implies $\delta G(\lambda, \ell - e_1 + e_2) > 0$.
Proposition~\ref{prop:monotonicity-ell}
then follows from
an induction argument that we omit.
The following lemma,
which follows directly
from~\eqref{eq:barpi},
will be useful.

\begin{lemma} \label{lem:recbarpi}
	For each $x \in \N^2$ such that $x_2 \ge 1$, we have
	\begin{equation*}
		\bar\pi(x)
		= \frac{x_1 + 1}{x_2} \frac{\mu_2}{\mu_1}
		\bar\pi(x + e_1 - e_2).
	\end{equation*}
\end{lemma}

Let $\ell \in \N^2$ such that
$\ell_1 \ge 1$, $\ell_2 \ge 1$, and
$\delta G(\lambda, \ell) > 0$.
By applying~\eqref{eq:deltaG}
to $k = \ell - e_1 + e_2$
and removing a (positive) term
from the first sum, we obtain
\begin{align*}
	\delta G(\lambda, k)
	&> \sum_{x_2 = 1}^{\ell_2} \bar\pi(\ell_1, x_2)
	- \sum_{x_1 = 0}^{\ell_1 - 1} \bar\pi(x_1, \ell_2 + 1).
\end{align*}

Applying Lemma~\ref{lem:recbarpi} to each term, using the fact that
$x_2 < \ell_2 + 1$ in the first sum
and $x_1 < \ell_1$ in the second sum,
and making a change of variable yields
\begin{align*}
	\frac{\ell_2 + 1}{\ell_1 + 1}
	\frac{\mu_1}{\mu_2}
	\delta G(\lambda, k)
	> \sum_{y_2 = 0}^{\ell_2 - 1}
	\bar\pi(\ell_1 + 1, y_2)
	- \sum_{y_1 = 1}^{\ell_1}
	\bar\pi(y_1, \ell_2).
\end{align*}
Lastly, we add a (positive) term
in the last sum
and apply~\eqref{eq:deltaG},
so as to obtain
\begin{align*}
	\frac{\ell_2 + 1}{\ell_1 + 1}
	\frac{\mu_1}{\mu_2}
	\delta G(\lambda, k)
	>
	\delta G(\ell).
\end{align*}
Therefore,
$\delta G(\lambda, \ell) > 0$
implies that
$\delta G(\lambda, k) > 0$.

\section{Proof of Proposition~\ref{prop:monotonicity-lambda}}
\label{app:monotonicity-lambda}

\theoremstyle{plain}
\newtheorem*{prop:monotonicity-lambda}
{Proposition~\ref{prop:monotonicity-lambda}}
\begin{prop:monotonicity-lambda}
	Let $\lambda_* > 0$
	and $\ell \in \N^2$
	such that $\ell_2 \ge 1$.
	If $\beta(\lambda_*, \ell + e_1 - e_2)
	< \beta(\lambda_*, \ell)$,
	then $\beta(\lambda, \ell + e_1 - e_2)
	< \beta(\lambda, \ell)$
	for each $\lambda \in (0, \lambda_*)$.
\end{prop:monotonicity-lambda}

\noindent
We will prove equivalently that,
for each $\lambda_* > 0$
and $\ell \in \N^2$
such that $\ell_2 \ge 1$,
$\delta G(\lambda_*, \ell) > 0$
implies $\delta G(\lambda, \ell) > 0$
for each $\lambda \in (0, \lambda_*)$.
The following lemma will be useful.

\begin{lemma} \label{lem:monotonicity-lambda}
	Let $\ell \in \N^2$
	such that $\ell_2 \ge 1$.
	There is a sequence
	$\{c_n\}_{n \in \{\min(\ell_1 + 1, \ell_2),
		\ldots, \ell_1 + \ell_2\}}$
	such that
	\begin{align} \label{eq:difference-c}
		\delta G(\lambda, \ell)
		= \sum_{n = \min(\ell_1 + 1, \ell_2)}
		^{\ell_1 + \ell_2}
		c_n \frac1{\lambda^n},
		\quad \lambda > 0.
	\end{align}
	The sequence
	$\{c_n\}_{n \in \{\min(\ell_1 + 1, \ell_2),
		\ldots, \ell_1 + \ell_2\}}$
	depends on the
	buffer lengths $\ell_1$ and $\ell_2$
	and service rates $\mu_1$ and $\mu_2$
	but not on the arrival rate $\lambda$,
	and satisfies one of
	the following conditions:
	\begin{enumerate}[(1)]
		\item \label{case:positive}
		$c_n > 0$ for each
		$n \in \{\min(\ell_1 + 1, \ell_2),
		\ldots, \ell_1 + \ell_2\}$,
		\item \label{case:negative}
		$c_n < 0$ for each
		$n \in \{\min(\ell_1 + 1, \ell_2),
		\ldots, \ell_1 + \ell_2\}$, or
		\item \label{case:mixed}
		there is
		$n^* \in \{\min(\ell_1 + 1, \ell_2),
		\ldots, \ell_1 + \ell_2\}$ such that
		$c_n < 0$
		for each
		$n \in \{ \min(\ell_1 + 1, \ell_2), \ldots, n^* - 1 \}$,
		$c_{n^*} \ge 0$,
		and $c_n > 0$
		for each
		$n \in \{ n^* + 1, \ldots, \ell_1 + \ell_2 \}$.
	\end{enumerate}
\end{lemma}

\begin{figure*}[b]
	\begin{subequations}
		\newcommand{\subequationsformat}{\theparentequation.\arabic{equation}}
		\begin{align}
			\label{eq:cn-1}
			c_n &= \begin{cases}
				\displaystyle
				\binom{n}{\ell_1 + 1}
				{\mu_1}^{\ell_1 + 1}
				{\mu_2}^{n - \ell_1 - 1},
				&n \in \{\ell_1 + 1, \ldots, \ell_2 - 1\}, \\[.3cm]
				\displaystyle
				\binom{n}{\ell_1 + 1}
				{\mu_1}^{\ell_1 + 1}
				{\mu_2}^{n - \ell_1 - 1}
				-
				\binom{n}{\ell_2}
				{\mu_1}^{n - \ell_2}
				{\mu_2}^{\ell_2},
				&n \in \{\ell_2, \ldots, \ell_1 + \ell_2\}.
			\end{cases} \\
			\label{eq:cn-2}
			c_n &= \begin{cases}
				\displaystyle
				- \binom{n}{\ell_2}
				{\mu_1}^{n - \ell_2}
				{\mu_2}^{\ell_2},
				&n \in \{\ell_2, \ldots,
				\ell_1\}, \\[.3cm]
				\displaystyle
				\binom{n}{\ell_1 + 1}
				{\mu_1}^{\ell_1 + 1} {\mu_2}^{n - \ell_1 - 1}
				-
				\binom{n}{\ell_2}
				{\mu_1}^{n - \ell_2}
				{\mu_2}^{\ell_2},
				&n \in \{\ell_1 + 1, \ldots,
				\ell_1 + \ell_2\}.
			\end{cases}
		\end{align}
	\end{subequations}
	\caption{Definitions of the sequence
		$\{c_n\}_{n \in \{\min(\ell_1 + 1, \ell_2),
			\ldots, \ell_1 + \ell_2\}}$
		depending on the value
		of $\ell = (\ell_1, \ell_2)$.}
	\label{fig:cn}
\end{figure*}

\begin{proof}[Proof of the lemma]
	As in the proof
	of Theorem~\ref{theo:overloaded},
	we use~\eqref{eq:difference-n}
	and distinguish two cases depending
	on the values of $\ell_1$ and~$\ell_2$.
	
	\paragraph*{Case 1
		($\ell_1 + 1 \le \ell_2$)}
	
	We can rewrite~\eqref{eq:difference-n}
	as~\eqref{eq:difference-c},
	where the sequence
	$\{c_n\}_{n \in \{\min(\ell_1 + 1, \ell_2),
		\ldots, \ell_1 + \ell_2\}}$
	is given
	by~\eqref{eq:cn-1}
	in \figurename~\ref{fig:cn}.
	It follows immediately that
	$c_n > 0$ for each
	$n \in \{\ell_1 + 1, \ldots, \ell_2 - 1\}$.
	Additionally, we verify that
	$c_n > 0$ for each
	$n \in \{\ell_2, \ldots, \ell_1 + \ell_2\}$
	by calculating the ratio between
	the first and second term of the subtraction
	and showing,
	thanks to the assumption $\mu_1 > \mu_2$,
	that this ratio is
	(strictly) larger than one.
	Therefore, the sequence
	satisfies Condition~\ref{case:positive}.

	\paragraph*{Case 2
		($\ell_1\ge \ell_2$)}
	
	We can rewrite~\eqref{eq:difference-n}
	as~\eqref{eq:difference-c},
	where the sequence
	$\{c_n\}_{n \in \{\min(\ell_1 + 1, \ell_2),
		\ldots, \ell_1 + \ell_2\}}$
	is given
	by~\eqref{eq:cn-2}
	in \figurename~\ref{fig:cn}.
	It follows directly that
	$c_n < 0$ for
	$n \in \{\ell_2, \ldots, \ell_1\}$.
	For
	$n \in \{\ell_1 + 1, \ldots, \ell_1 + \ell_2\}$,
	$c_n$ is of the same sign as $g_n$,
	where
	\begin{align*}
		g_n = \log\left(
		\frac
		{\binom{n}{\ell_1 + 1}
			{\mu_1}^{\ell_1 + 1}
			{\mu_2}^{n - \ell_1 - 1}}
		{\binom{n}{\ell_2}
			{\mu_1}^{n - \ell_2}
			{\mu_2}^{\ell_2}}
		\right).
	\end{align*}
	We conclude that the sequence
	$\{c_n\}_{n \in \{\min(\ell_1 + 1, \ell_2),
		\ldots, \ell_1 + \ell_2\}}$
	satisfies Conditions~\ref{case:negative}
	or~\ref{case:mixed}
	by observing that
	$g_{\ell_1 + \ell_2 + 1} = 0$,
	and then showing that the sequence
	$\{g_n\}_{n \in \{\ell_1 + 1, \ldots,
		\ell_1 + \ell_2 + 1\}}$
	is strictly concave
	in the sense that
	$g_n - g_{n+1} > g_{n+1} - g_{n+2}$
	for each $n \in \{\ell_1 + 1,
	\ldots, \ell_1 + \ell_2 - 1\}$.
	The details of the calculation
	are omitted due to space constraints.
	
	\paragraph*{Conclusion}
	
	The sequence
	$\{c_n\}_{n \in \{\min(\ell_1 + 1, \ell_2),
		\ldots, \ell_1 + \ell_2\}}$
	satisfies Conditions~\ref{case:positive},
	\ref{case:negative}, or \ref{case:mixed}
	of the proposition.
\end{proof}

We now prove
Proposition~\ref{prop:monotonicity-lambda}.
Let $\ell \in \N^2$
such that $\ell_2 \ge 1$.
Using notation from
Lemma~\ref{lem:monotonicity-lambda},
we distinguish two cases.
If the sequence
$\{c_n\}_{n \in \{\min(\ell_1 + 1, \ell_2),
	\ldots, \ell_1 + \ell_2\}}$
satisfies
Conditions \ref{case:positive} or \ref{case:negative},
the result is immediate
because the sign of $\delta G(\lambda, \ell)$
does not depend on $\lambda$.
Now assume that the sequence
$\{c_n\}_{n \in \{\min(\ell_1 + 1, \ell_2),
	\ldots, \ell_1 + \ell_2\}}$
satisfies Condition~\ref{case:mixed}.
If $n^* = \ell_1 + \ell_2$
and $c_{n^*} = 0$,
the conclusion is the same
as under Condition~\ref{case:negative}.
Otherwise, we study the variations
of the function
$f: \lambda \mapsto
\delta G(\lambda, \ell)$
on $(0, +\infty)$.
This function is infinitely differentiable
on this interval and,
for each $\lambda > 0$, we have
\begin{align*}
	f'(\lambda)
	&= - \sum_{n = \min(\ell_1 + 1, \ell_2)}^{\ell_1 +\ell_2}
	n c_n \frac1{\lambda^{n+1}}, \\
	&= - \sum_{n = \min(\ell_1 + 1, \ell_2)}^{n^* - 1}
	n c_n \frac1{\lambda^{n+1}}
	- \sum_{n = n^*}^{\ell_1 + \ell_2}
	(n^* - 1) c_n \frac1{\lambda^{n+1}}
	- \sum_{n = n^*}^{\ell_1 + \ell_2}
	(n - (n^* - 1))
	c_n \frac1{\lambda^{n+1}}, \\
	&\le - (n^* - 1) \left(
	\sum_{n = \min(\ell_1 + 1, \ell_2)}^{\ell_1 + \ell_2}
	c_n \frac1{\lambda^{n+1}}
	\right)
	- \sum_{n = n^*}^{\ell_1 + \ell_2}
	(n - (n^* - 1)) c_n \frac1{\lambda^{n+1}}, \\
	&= - (n^* - 1) f(\lambda)
	- \sum_{n = n^*}^{\ell_1 + \ell_2}
	(n - (n^* - 1)) c_n \frac1{\lambda^{n+1}}, \\
	&<  - (n^* - 1) f(\lambda).
\end{align*}
Hence, if there is $\lambda > 0$
such that $f(\lambda) \le 0$,
then $f'(\lambda) < 0$.
Lemma~\ref{lem:monotonicity-lambda}
also implies
$\lim_{\lambda \to 0^+} f(\lambda)
= +\infty$.
Therefore, either
$f$ is positive on $(0, +\infty)$,
or there is $\lambda_0 > 0$ such that
$f$ is positive on $(0, \lambda_0)$
and negative on $(\lambda_0, +\infty)$,
with $f(\lambda_0) = 0$.
Both cases lead to the conclusion.

\end{document}